\documentclass[11pt]{article}
\usepackage{amsmath,amssymb,amsthm}
\usepackage{enumerate}
\usepackage{hyperref}
\usepackage{cleveref}
\usepackage{mathabx}
\usepackage{blkarray}
\usepackage{tikz}
\usepackage{authblk}
\usepackage{thmtools}
\usepackage{thm-restate}
\usepackage{lineno}
\usepackage{tabu}
\usepackage{makecell}

\usetikzlibrary{backgrounds}
\usetikzlibrary{shapes}
\usetikzlibrary{positioning}
\hypersetup{
	pdftitle = {The proper conflict-free k-coloring problem
	and the odd k-coloring problem
	are NP-complete on bipartite graphs},
	pdfauthor = {Jungho Ahn, Seonghyuk Im, and Sang-il Oum}
}
\newcommand\abs[1]{\lvert #1\rvert}
\newtheorem{THM}{Theorem}[section]
\newtheorem*{THMplanar}{Theorem~\ref{thm:mainplanar}}
\newtheorem*{THMbipartite}{Theorem~\ref{thm:mainbipartite}}
\newtheorem*{THModdbipartite}{Theorem~\ref{thm:mainoddbipartite}}
\newtheorem{LEM}[THM]{Lemma}

\newtheorem{PROP}[THM]{Proposition}

\newtheorem{PROB}{Problem}
\theoremstyle{remark}

\theoremstyle{definition}

\newcommand{\pcf}{\chi_{\mathrm{pcf}}}
\newcommand{\och}{\chi_{\mathrm{odd}}}

\usepackage{boxedminipage}

\begin{document}
\date{\today}
\title{The proper conflict-free $k$-coloring problem and the odd $k$-coloring problem are NP-complete on bipartite graphs}
\author[2,1]{Jungho Ahn\thanks{Supported by the Institute for Basic Science (IBS-R029-C1).}} 
\author[2,3]{Seonghyuk Im\thanks{Supported by the Institute for Basic Science (IBS-R029-C4) and the POSCO Science Fellowship of POSCO TJ Park Foundation.}}
\author[$\ast$1,2]{Sang-il Oum}
\affil[1]{Discrete Mathematics Group, Institute for Basic Science (IBS), Daejeon, South~Korea}
\affil[2]{Department of Mathematical Sciences, KAIST, Daejeon,~South~Korea}
\affil[3]{Extremal Combinatorics and Probability Group, Institute~for~Basic~Science~(IBS), Daejeon, South~Korea}
\affil[ ]{\small \textit{Email addresses:} \texttt{junghoahn@kaist.ac.kr}, \texttt{seonghyuk@kaist.ac.kr}, \texttt{sangil@ibs.re.kr}}

\maketitle

\begin{abstract}
	A proper coloring of a graph is \emph{proper conflict-free} if every non-isolated vertex $v$ has a neighbor whose color is unique in the neighborhood of~$v$.
	A proper coloring of a graph is \emph{odd} if for every non-isolated vertex~$v$, there is a color appearing an odd number of times in the neighborhood of $v$.
	For an integer $k$, the \textsc{PCF $k$-Coloring} problem asks whether an input graph admits a proper conflict-free $k$-coloring and the \textsc{Odd $k$-Coloring} asks whether an input graph admits an odd $k$-coloring.
	We show that for every integer $k\geq3$, both problems are NP-complete, even if the input graph is bipartite.
	Furthermore, we show that the \textsc{PCF $4$-Coloring} problem is NP-complete when the input graph is planar.
\end{abstract}

\section{Introduction}\label{sec:intro}

Fabrici, Lu\v{z}ar, Rindo\v{s}ov\'a, and Sot\'{a}k~\cite{Fabrici2022} introduced the proper conflict-free coloring of graphs.
For a positive integer $k$, a \emph{proper $k$-coloring} of a graph is a function that maps each vertex to one of the $k$ colors such that adjacent vertices receive distinct colors.
A \emph{proper conflict-free $k$-coloring} of a graph is a proper $k$-coloring such that every non-isolated vertex~$v$ has a neighbor whose color is unique in the neighborhood of $v$.
We abbreviate it as a \emph{PCF $k$-coloring}.
The \emph{proper conflict-free chromatic number} or \emph{PCF chromatic number} of~$G$, denoted by $\pcf(G)$, is the smallest integer~$k$ such that $G$ admits a PCF $k$-coloring.

Petru\v{s}evski and \v{S}krekovski \cite{Petrusevski2021} introduced the odd coloring of graphs. 
An \emph{odd $k$-coloring} of a graph is a proper $k$-coloring such that for every non-isolated vertex $v$ of~$G$, there is a color appearing an odd number of times in the neighborhood of $v$.
The \emph{odd chromatic number} of~$G$, denoted by $\och(G)$, is the smallest integer~$k$ such that $G$ admits an odd $k$-coloring.

By definition, every proper conflict-free coloring is an odd coloring, so $\och(G)\leq\pcf(G)$.
The gap between the chromatic number and the odd chromatic number of a graph could be arbitrary large.
For instance, let $G$ be the graph obtained from the complete graph on $n$ vertices by replacing each edge with a length-$2$ path.
Since $G$ is bipartite, it has the chromatic number $2$, but it is known that $\och(G)=\pcf(G)=n$~\cite{Caro2022}.
Both colorings have been investigated actively~\cite{Caro2022a,Caro2022,Cho2022a,Cho2022,Cranston2022,Fabrici2022,Hickingbotham2022,Liu2022,Petrusevski2021,Petrusevski2021a,Qi2022}.

For an integer $k$, the \textsc{PCF $k$-Coloring} problem or the \textsc{Odd $k$-Coloring} problem asks whether an input graph admits a PCF $k$-coloring or an odd $k$-coloring, respectively.
It is readily seen that for a graph $G$,
\begin{itemize}
	\item	$\pcf(G)\le 2$ if and only if $G$ has the maximum degree at most $1$, and 
	\item	$\och(G)\le 2$ if and only if $G$ is bipartite and $\deg_G(v)$ is odd or zero for every vertex $v$ of~$G$.
\end{itemize}
Thus, for $k\leq2$, both the \textsc{PCF $k$-Coloring} problem and the \textsc{Odd $k$-Coloring} problem are polynomial-time solvable.

Caro, Petru\v{s}evski, and \v{S}krekovski~\cite{Caro2022a,Caro2022} showed that the following reductions lead to the NP-hardness of both problems.

\begin{LEM}[Caro, Petru\v{s}evski, and \v{S}krekovski~\cite{Caro2022}]\label{lem:red1}
	For a graph $G$, let~$H$ be the graph obtained from $G$ by adding a pendant vertex to every vertex of~$G$.
	Then $\chi(G)\le \pcf(H)\le \chi(G)+1$.
\end{LEM}

\begin{LEM}[Caro, Petru\v{s}evski, and \v{S}krekovski~\cite{Caro2022}]\label{lem:red2}
	For a graph $G$, let~$H$ be the graph obtained from $G$ by adding one vertex adjacent to all other vertices.
	Then $\chi(G)+1\le \pcf(H)\le \chi(G)+2$.
\end{LEM}

\begin{LEM}[Caro, Petru\v{s}evski, and \v{S}krekovski~\cite{Caro2022a}]\label{lem:redodd}
	For a graph $G$, let~$H$ be the graph obtained from $G$ by adding a pendant vertex to every vertex of~$G$ having even degree.
	Then $\chi(G)=\och(H)$.
\end{LEM}

Here is a variation of Lemma~\ref{lem:red2}.
We omit its easy proof.

\begin{LEM}
	For a graph $G$,
	let $H$ be the graph obtained from $G$ by adding two new adjacent vertices and making them adjacent to all other vertices. Then $\chi(G)+2= \pcf(H)$.
\end{LEM}

For $k\geq3$, above lemmas can be combined with the NP-hardness of deciding $\chi(G)\le k$, or the NP-hardness of deciding whether $\chi(G)\leq k$ or $\chi(G)\geq k+2$ shown by Khanna, Linial, and Safra~\cite{Khanna2000}.
That is how Caro, Petru\v{s}evski, and \v{S}krekovski~\cite{Caro2022a,Caro2022} showed that 
it is NP-complete to decide whether \[ \pcf(G)\le k \text{ for } k\ge 4, \] and it is NP-complete to decide whether \[ \och(G)\le k \text{ for } k\ge 3.\] 
As these reductions require $\chi(G)\ge 3$ to be NP-hard, 
we cannot use the above reductions to prove that our problems are NP-complete on bipartite graphs.

A graph is \emph{subcubic} if every vertex has degree at most $3$.
We found a reference implying that the \textsc{PCF $3$-Coloring} problem is NP-complete on subcubic bipartite planar graphs.
In 2009, Li, Yao, Zhou, and Broersma~\cite{LYZB2009} showed that it is NP-complete to decide whether a $2$-connected subcubic bipartite planar graph admits a proper $3$-coloring such that for every vertex~$v$ of $G$, its neighbors have at least two colors.
For a subcubic graph, such a coloring is precisely the PCF $3$-coloring, and therefore it implies that it is NP-complete to decide whether $\pcf(G)\le3$ on subcubic bipartite planar graphs.
In addition, from their reduction, it can be easily seen that it is NP-complete to decide whether $\och(G)\leq3$ on subcubic bipartite planar graphs.

We prove that it is NP-complete to decide whether $\pcf(G)\le k$ for $k\ge 3$, even if $G$ is bipartite.

\begin{THMbipartite}
	For every integer $k\geq3$, it is NP-complete to decide whether a graph $G$ admits 
	a PCF $k$-coloring, even if $G$ is bipartite. 
\end{THMbipartite}

We also prove that it is NP-complete to decide whether $\och(G)\le k$ for $k\ge 3$, even if $G$ is bipartite.

\begin{THModdbipartite}
	For every integer $k\geq3$, it is NP-complete to decide whether a graph $G$ admits an odd $k$-coloring, even if $G$ is bipartite.
\end{THModdbipartite}

In addition, we show that it is NP-complete to decide whether $\pcf(G)\le4$, even if $G$ is planar. 

\begin{THMplanar}
	It is NP-complete to decide whether a graph $G$ admits a PCF $4$-coloring, even if $G$ is planar.
\end{THMplanar}

Fabrici, Lu\v{z}ar, Rindo\v{s}ov\'a, and Sot\'{a}k \cite{Fabrici2022} presented a planar graph having the PCF chromatic number $6$ and showed that every planar graph admits a PCF $8$-coloring.
Thus, it remains as an open problem to determine the computational complexity of deciding whether $\pcf(G)\le k$ on planar graphs~$G$ when $k\in \{5,6,7\}$. 

We organize this paper as follows.
In Section~\ref{sec:prelim}, we introduce some terminologies in graph theory.
In Section~\ref{sec:bip}, we prove Theorems~\ref{thm:mainbipartite} and~\ref{thm:mainoddbipartite}.
In Section~\ref{sec:planar}, we prove Theorem~\ref{thm:mainplanar}.
In Section~\ref{sec:end}, we conclude the paper with some open problems.

\section{Preliminaries}\label{sec:prelim}

In this paper, all graphs are simple and finite.
For a positive integer $k$, let~$K_k$ be the complete graph on $k$ vertices and $[k]:=\{1,2,\ldots,k\}$.
A \emph{pendant vertex} is a vertex of degree $1$.
A \emph{cligue} in a graph $G$ is a set of pairwise adjacent vertices of $G$.
The \emph{$k$-subdivision} of $G$, denoted by $\mathrm{sub}_k(G)$, is the graph obtained from~$G$ by replacing each edge with a path of length $k+1$.
A graph $G$ is \emph{$k$-connected} if it has more than $k$ vertices and $G\setminus X$ is connected for every set $X\subseteq V(G)$ of size less than $k$.

A \emph{plane} graph is a planar graph embedded into $\mathbb{R}^2$ without crossings of edges.
It is well known that in every $2$-connected plane graph, each face is bounded by a cycle; see \cite[Proposition~4.2.6]{Diestel}.

We will use the following simple lemmas.

\begin{LEM}\label{lem:degree2}
	In any odd $k$-coloring of a graph,
	the neighbors of a degree-$2$ vertex have distinct colors.
\end{LEM}
\begin{proof}
	It is trivial from the definition of an odd $k$-coloring.
\end{proof}

\begin{LEM}\label{lem:subdivision_lowerbound}
	$\chi(G)\le \och(\mathrm{sub}_1(G))\le \pcf(\mathrm{sub}_1(G))$ for every graph $G$.
\end{LEM}
\begin{proof}
	By Lemma~\ref{lem:degree2}, every odd $k$-coloring of $\mathrm{sub}_1(G)$ induces a proper $k$-coloring of $G$.
\end{proof}

\section{NP-completeness on bipartite graphs}\label{sec:bip}

We now restate and show Theorems~\ref{thm:mainbipartite} and ~\ref{thm:mainoddbipartite}.

\begin{THM}\label{thm:mainbipartite}
	For every integer $k\geq3$, it is NP-complete to decide whether a graph $G$ admits 
	a PCF $k$-coloring, even if $G$ is bipartite. 
\end{THM}

\begin{THM}\label{thm:mainoddbipartite}
	For every integer $k\geq3$, it is NP-complete to decide whether a graph $G$ admits an odd $k$-coloring, even if $G$ is bipartite.
\end{THM}

The following lemma immediately implies Theorems \ref{thm:mainbipartite} and \ref{thm:mainoddbipartite} for $k\ge 5$ because it is NP-complete to decide whether $\chi(G)\le k$ for a graph $G$~\cite{GJ1979}.

\begin{LEM}\label{lem:subdivision_largechromatic}
	For every graph $G$, 
	\[ \chi(G)\le \och(\mathrm{sub}_1(G))\le \pcf(\mathrm{sub}_1(G))\le
	\max(\chi(G),5).\]
\end{LEM}
\begin{proof}
	By Lemma~\ref{lem:subdivision_lowerbound}, the first two inequalities hold. 
	Let $k:=\max(\chi(G),5)$ and let $H:=\mathrm{sub}_1(G)$.
	We may assume that $G$ has no isolated vertices.
	It remains to show that $\pcf(H)\leq k$.

	Let $c'$ be a proper $k$-coloring of $G$.
	Let $X$ be a maximal subset of $V(H)$ containing $V(G)$ such that 
	there exists a PCF $k$-coloring $c$ of $H[X]$ extending~$c'$.
	We claim that $X=V(H)$. If not, then $H$ has a vertex $v$ of degree $2$ not in $X$. Let $x$, $y$ be the neighbors of $v$. 
	We say a vertex $w$ uses a color $i$ if $c(w)=i$ or $i$ is the color appearing uniquely in the neighbors of $w$ in $H[X]$.
	Then each of $x$ or $y$ uses at most $2$ colors.
	Since $k\ge 5$, there is a color~$i$ not used in $x$ or $y$.
	Then we can extend $c$ to a PCF $k$-coloring of $H[X\cup\{v\}]$ by making $c(v):=i$, contradicting the assumption that $X$ is chosen to be maximal. 
\end{proof}

The theorem of Li, Yao, Zhou, and Broersma~\cite{LYZB2009} implies that it is NP-complete to decide whether a bipartite graph admits a PCF $3$-coloring, proving Theorem~\ref{thm:mainbipartite} for $k=3$.
Their proof also implies that it is NP-complete to decide whether a bipartite graph admits an odd $3$-coloring, proving Theorem~\ref{thm:mainoddbipartite} for $k=3$.
The following proposition states that the problem of deciding $\pcf(G)\le3$ or $\och(G)\le 3$
can be reduced to the problem of deciding $\pcf(G)\le 4$ or $\och(G)\le 4$, respectively, 
proving the theorems for $k=4$.

\begin{PROP}\label{lem:bipartite4}
	For a bipartite graph $G$, one can construct a bipartite graph~$\widetilde{G}$ in polynomial time such that 
	\begin{itemize}
		\item $\pcf(G)\leq3$ if and only if $\pcf(\widetilde{G})\leq4$, and 
		\item $\och(G)\leq3$ if and only if $\och(\widetilde{G})\leq4$.
	\end{itemize}
\end{PROP}

To prove Proposition~\ref{lem:bipartite4}, we will use the following lemma.

\begin{LEM}\label{lem_complete_subdivision}
	Let $G$ be the $1$-subdivision of $K_4$, and $x$ be a degree-$3$ vertex of $G$. 
	Then $G$ has a PCF $4$-coloring~$c$ such that 
	\begin{enumerate}[(a)]
		\item all neighbors of $x$ have distinct colors,
		\item each degree-$3$ vertex $y\neq x$ has a neighbor $z$ non-adjacent to $x$
			such that $c(z)\neq c(x)$ and the color of $z$ is unique in the colors of neighbors of $y$.
	\end{enumerate}
\end{LEM}
\begin{proof}
	Let $v_1,\ldots,v_4:=x$ be the degree-$3$ vertices of $G$ and for $1\le i<j\le 4$, let $s_{ij}$ be the degree-$2$ vertex adjacent to both $v_i$ and $v_j$.
	For each vertex $v$ of $G$, let 
	\[
		c(v):=
		\begin{cases}
			1	& \text{if $v\in\{v_1,s_{34}\}$},\\
			2	& \text{if $v\in\{v_2,s_{13},s_{14}\}$},\\
			3	& \text{if $v\in\{v_3,s_{12},s_{24}\}$},\\
			4	& \text{if $v\in\{v_4,s_{23}\}$}. 
		\end{cases} 
	\]
	Then $c$ is a desired PCF $4$-coloring of $G$.
\end{proof}

We now prove Proposition~\ref{lem:bipartite4}.

\begin{proof}[Proof of Proposition~\ref{lem:bipartite4}]
	Let $(A, B)$ be a bipartition of~$G$.
	If $\abs{V(G)}\leq3$, then we can take $\widetilde{G}:=G$.
	Thus, we may assume that $\abs{V(G)}>3$.
	By symmetry, we may further assume that $\abs{B}\geq 2$.

	For positive integers $n$ and $m$, let $G_{n,m}$ be the graph whose vertex set is 
	$\{ a_1,a_2,\ldots,a_{2n}, \alpha_1,\alpha_2,\alpha_3,
	b_1,b_2,\ldots,b_{2m},\beta_1,\beta_2,\beta_3\}$
	such that 
	\begin{enumerate}[(i)]
		\item	for each $1\le i\le 2n$, $a_i$ is adjacent to all of $\alpha_1$, $\alpha_2$, and $\alpha_3$, 
		\item	for each $1\le j\le 2m$, $b_j$ is adjacent to all of $\beta_1$, $\beta_2$, and $\beta_3$, and 
		\item	both $\{\alpha_1,\alpha_2,\alpha_3\}$ and $\{\beta_1,\beta_2,\beta_3\}$ are cliques.
	\end{enumerate}
	Let $\widetilde{G}$ be the graph obtained from the disjoint union of $G$ and $\mathrm{sub}_1(G_{\abs{A},\abs{B}})$
	by the following operations. 
	\begin{itemize}
		\item For each $1\le i\le \abs{A}$, we add two edges from the $i$-th vertex of $A$ to $a_{2i-1}$ and $a_{2i}$.
		\item For each $1\le j\le \abs{B}$, we add two edges from the $j$-th vertex of $B$ to $b_{2j-1}$ and $b_{2j}$.
		\item We add three edges $\alpha_1b_1$, $\alpha_2b_2$, and $\alpha_3b_3$.
	\end{itemize}
	See Figure~\ref{fig:example} for an illustration.
	Note that $\widetilde{G}$ is bipartite and can be constructed in polynomial time.

	\begin{figure}\label{fig:example}
		\centering
		\tikzstyle{v}=[circle, draw, solid, fill=black, inner sep=0pt, minimum width=3pt]
		\begin{tikzpicture}[scale=0.5]
			\draw (-5,2) arc (120:60:14);
			\draw (-5,2) arc (120:60:10);
			\draw (-5,2) node[v,label=$v$](v){};
			\draw (5,2) node[v,label=$u$](u){};
			\draw (9,2) node[v,label=$w$](w){};
			\draw (-6,0) node[v,label=left:$a_1$](a1){};
			\draw (-4,0) node[v,label=right:$a_2$](a2){};
			\draw (4,0) node[v,label=right:$b_1$](b1){};
			\draw (6,0) node[v,label=right:$b_2$](b2){};
			\draw (8,0) node[v,label=right:$b_3$](b3){};
			\draw (10,0) node[v,label=right:$b_4$](b4){};
			\draw (a1)--(v)--(a2);
			\draw (b1)--(u)--(b2);
			\draw (b3)--(w)--(b4);
			\draw[dash pattern=on 2mm off 0.5mm] (b1)--(7,-5)--(b2);
			\draw[dash pattern=on 2mm off 0.5mm] (b1)--(5,-6.5)--(b2);
			\draw[dash pattern=on 2mm off 0.5mm] (b1)--(9,-6.5)--(b2);
			\draw[dash pattern=on 2mm off 0.5mm] (b3)--(7,-5)--(b4);
			\draw[dash pattern=on 2mm off 0.5mm] (b3)--(5,-6.5)--(b4);
			\draw[dash pattern=on 2mm off 0.5mm] (b3)--(9,-6.5)--(b4);
			
			\draw (-5,-4.5) node[v,label=below:$\alpha_1$](aa1){};
			\draw (-7,-6.5) node[v,label=left:$\alpha_2$](aa2){};
			\draw (-3,-6.5) node[v,label=right:$\alpha_3$](aa3){};
			\draw (7,-5) node[v,label=below:$\beta_1$](bb1){};
			\draw (5,-6.5) node[v,label=left:$\beta_2$](bb2){};
			\draw (9,-6.5) node[v,label=right:$\beta_3$](bb3){};
			\draw[dash pattern=on 2mm off 0.5mm] (aa1)--(aa2)--(aa3)--(aa1);
			\draw[dash pattern=on 2mm off 0.5mm] (bb1)--(bb2)--(bb3)--(bb1);
			\draw[dash pattern=on 2mm off 0.5mm] (a1)--(aa1)--(a2);
			\draw[dash pattern=on 2mm off 0.5mm] (a1)--(aa2)--(a2);
			\draw[dash pattern=on 2mm off 0.5mm] (a1)--(aa3)--(a2);
			\draw (aa1)--(b1);
			\draw (aa2)--(b2);
			\draw (aa3)--(b3);
		\end{tikzpicture}
		\caption{The graph $\widetilde{P}$ for a path $P:=uvw$, where dashed lines are length-$2$ paths.}
		\label{fig:path2}
	\end{figure}
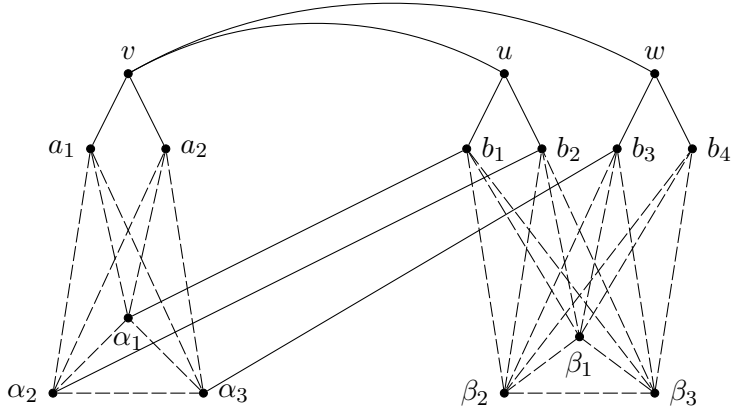

	First, let us show that if $\pcf(\widetilde{G})\le 4$, then $\pcf(G)\le 3$.
	Suppose that $\widetilde{G}$ has a PCF $4$-coloring $c$.
	We show that the restriction $c'$ of $c$ on $V(G)$ is a PCF $3$-coloring of~$G$.
	For each $i\in[2\abs{A}]$, by Lemma~\ref{lem:degree2}, $c(\alpha_1)$, $c(\alpha_2)$, $c(\alpha_3)$, and $c(a_i)$ are pairwise distinct, and therefore $c(a_1)=c(a_2)=\cdots=c(a_{2\abs{A}})$.
	Similarly, $c(b_1)=c(b_2)=\cdots=c(a_{2\abs{B}})$.
	For each $\ell\in[3]$, since $\alpha_\ell$ and $b_\ell$ are adjacent in $\widetilde{G}$, we have that $c(\alpha_\ell)\neq c(b_\ell)$.
	Hence, 
	\[ 
		c(a_1)=c(a_2)=\cdots=c(a_{2\abs{A}})=c(b_1)=c(b_2)=\cdots=c(a_{2\abs{B}}).
	\] 
	Thus, $c'$ uses at most $3$ colors.
	For each vertex $v$ of~$G$, $N_{\widetilde{G}}(v)\setminus N_G(v)$ is exactly the set of two vertices $w$ and $w'$ with $c(w)=c(w')$.
	Thus, if $v$ has a neighbor $u$ in $\widetilde{G}$ such that $c(u)$ is unique in $N_{\widetilde{G}}(v)$, then $u\in V(G)$ and $c(u)$ is unique in $N_G(v)$.
	Hence, $c'$ is a PCF $3$-coloring of~$G$.
	Similarly, if $\och(\widetilde{G})\le 4$, then $\och(G)\le 3$.
 
	Now, let us show that if $\pcf(G)\le 3$, then $\pcf(\widetilde{G})\le 4$.
	Suppose that $G$ has a PCF $3$-coloring $c$.
	Let $c^*$ be a PCF $4$-coloring of $\mathrm{sub}_1(G_{\abs{A},\abs{B}})$
	obtained from Lemma~\ref{lem_complete_subdivision}
	such that all vertices in $\{a_1,a_2,\ldots,a_{2\abs{A}},b_1,b_2,\ldots,b_{2\abs{B}}\}$ have the same color by playing the role of $x$ in Lemma~\ref{lem_complete_subdivision}.
	We may assume that $c^*(a_1)=4$.

	For a vertex $v$ of $\widetilde{G}$, let $c'(v)= c(v)$ if $v\in V(G)$, and $c'(v)=c^*(v)$ if $v\in V(\mathrm{sub}_1(G_{\abs{A},\abs{B}}))$.
	We claim that $c'$ is a PCF $4$-coloring of $\widetilde{G}$.

	For each vertex $v$ of~$G$, all vertices in $N_{\widetilde{G}}(v)\setminus V(G)$ are colored by $4$
	and so it is easy to observe that $c'$ is a proper coloring.

	If $v\in V(G)$, then there are precisely two neighbors of $v$ not in $G$, both having the color $4$. 
	Since $c$ is a PCF $3$-coloring of $G$, 
	there is a color in $\{1,2,3\}$ appearing exactly once in the neighborhood of $v$ in $\widetilde{G}$.

	By Lemma~\ref{lem_complete_subdivision}(b), each vertex $v$ in $\{\alpha_1,\alpha_2,\alpha_3,\beta_1,\beta_2,\beta_3\}$ has a neighbor~$w$ non-adjacent to any vertex in $\{a_1,a_2,\ldots,a_{2\abs{A}},b_1,b_2,\ldots,b_{2\abs{B}}\}$ such
	that $c^*(w)\neq 4$ and the color of $w$ is unique in the neighborhood of~$v$ in $\mathrm{sub}_1(G_{\abs{A},\abs{B}})$. 
	This implies that such a vertex $v$ still satisfies the conflict-free condition in $\widetilde{G}$, because 
	$b_1$, $b_2$, $b_3$ have the color $4$.

	If $v$ is a degree-$2$ vertex of $\mathrm{sub}_1(G_{\abs{A},\abs{B}})$, then trivially $v$ satisfies the conflict-free condition in $\widetilde{G}$.
	If $v\in \{a_1,a_2,\ldots,a_{2\abs{A}},b_1,b_2,\ldots,b_{2\abs{B}}\}$, 
	then by Lemma~\ref{lem_complete_subdivision}(a),
	all three colors $1$, $2$, and $3$ appear in the neighborhood of~$v$ in $\mathrm{sub}_1(G_{\abs{A},\abs{B}})$.
	Since the degree of $v$ in $\widetilde{G}$ is at most two plus the degree of $v$ in $G_{\abs{A},\abs{B}}$, at least one of the colors $1$, $2$, and $3$ appears exactly once in the neighborhood of $v$ in $\widetilde{G}$.
	Hence, $c'$ is a PCF $4$-coloring of $\widetilde{G}$.
	Similarly, if $\och(G)\le 3$, then $\och(\widetilde{G})\le 4$.
\end{proof}

\section{NP-completeness on planar graphs}\label{sec:planar}

We now turn our attention to planar graphs.

\begin{THM}\label{thm:mainplanar}
	It is NP-complete to decide whether a graph $G$ admits a PCF $4$-coloring, even if $G$ is planar.
\end{THM}

The theorem of Li, Yao, Zhou, and Broersma~\cite{LYZB2009} implies that it is NP-complete to decide whether a $2$-connected planar graph admits a PCF $3$-coloring.
The following proposition combined with their theorem 
immediately implies Theorem~\ref{thm:mainplanar}.

For a plane graph $G$ and a face $f$ bounded by a cycle~$C$ of length~$k$,
attaching a \emph{tent} to $f$ is an operation to create a plane graph from $G$ by 
\begin{itemize}
	\item adding a cycle $v_1^fv_2^f\cdots v_{4k+2}^\ell v_1^f$ of length $4k+2$ inside $f$,
	\item attaching a pendant vertex $\ell_i^f$ to $v_i^f$ for each $i\in [4k+2]$,
	\item adding a \emph{center} $v^f$ and making it adjacent to all vertices in the cycle $v_1^fv_2^f\cdots v_{4k+2}^\ell v_1^f$, 
	\item adding a vertex $w^f$ inside the triangle $v^f v_1^f v_{4k+2}^f v^f$ and making it adjacent to all three vertices of the triangle, and
	\item making the $i$-th vertex of $C$ adjacent to both $v_{4i-2}^f$ and $v_{4i}^f$ for each $i\in[k]$.
\end{itemize}
See Figure~\ref{fig:Fgraph} for an illustration.
If we fix an ordering of the vertices of each face boundary of $G$, then the plane graph obtained from $G$ by attaching a tent to every face is unique up to isomorphism.

\begin{figure}
	\centering
	\tikzstyle{v}=[circle, draw, solid, fill=black, inner sep=0pt, minimum width=3pt]
	\begin{tikzpicture}[scale=0.9]
		\draw (0,0) circle (4);
		\draw (0,0) circle (2);
		\draw (90:4) node[v,label={$u_1$}](u1){};
		\draw (90-24*4:4) node[v,label=right:$u_2$](u2){};
		\draw (90-24*8:4) node[v,label=below:$u_3$](u3){};
		\draw (90-24*0:2) node[v,label={[xshift=2.5mm,yshift=-5.5mm]$v_3$}](v3){};
		\draw (90-24*1:2) node[v,label={[xshift=1.8mm,yshift=-6mm]$v_4$}](v4){};
		\draw (90-24*2:2) node[v,label={[xshift=0.3mm,yshift=-6.5mm]$v_5$}](v5){};
		\draw (90-24*3:2) node[v,label={[xshift=-1.5mm,yshift=-6mm]$v_6$}](v6){};
		\draw (90-24*4:2) node[v,label={[xshift=-2.5mm,yshift=-5mm]$v_7$}](v7){};
		\draw (90-24*5:2) node[v,label={[xshift=-3mm,yshift=-3.8mm]$v_8$}](v8){};
		\draw (90-24*6:2) node[v,label={[xshift=-3.3mm,yshift=-2.5mm]$v_9$}](v9){};
		\draw (90-24*7:2) node[v,label={[xshift=-3.4mm,yshift=-1.1mm]$v_{10}$}](v10){};
		\draw (90-24*8:2) node[v,label={[xshift=-2.3mm,yshift=-0.4mm]$v_{11}$}](v11){};
		\draw (90-24*9:2) node[v,label={[xshift=-0.8mm,yshift=0.5mm]$v_{12}$}](v12){};
		\draw (90-24*10:2) node[v,label={[xshift=1.5mm,yshift=0.6mm]$v_{13}$}](v13){};
		\draw (90-24*11:2) node[v,label={[xshift=-2.9mm,yshift=-0.4mm]$v_{14}$}](v14){};
		\draw (90-24*12:1.2) node[v,label={[xshift=-3mm,yshift=-1.5mm]$w$}](w){};
		\draw (90-24*13:2) node[v,label={[xshift=-3.2mm,yshift=-3.5mm]$v_1$}](v1){};
		\draw (90-24*14:2) node[v,label={[xshift=3.3mm,yshift=-4.45mm]$v_2$}](v2){};
		
		\draw (90-24*0:2.5) node[v,label=above:{$\ell_3$}](l3){};
		\draw (90-24*1:2.5) node[v,label={$\ell_4$}](l4){};
		\draw (90-24*2:2.5) node[v,label={$\ell_5$}](l5){};
		\draw (90-24*3:2.5) node[v,label=right:{$\ell_6$}](l6){};
		\draw (90-24*4:2.5) node[v,label=right:{$\ell_7$}](l7){};
		\draw (90-24*5:2.5) node[v,label=right:{$\ell_8$}](l8){};
		\draw (90-24*6:2.5) node[v,label=right:$\ell_9$](l9){};
		\draw (90-24*7:2.5) node[v,label=right:$\ell_{10}$](l10){};
		\draw (90-24*8:2.5) node[v,label=below:$\ell_{11}$](l11){};
		\draw (90-24*9:2.5) node[v,label=left:{$\ell_{12}$}](l12){};
		\draw (90-24*10:2.5) node[v,label=left:{$\ell_{13}$}](l13){};
		\draw (90-24*11:2.5) node[v,label=left:{$\ell_{14}$}](l14){};
		\draw (90-24*13:2.5) node[v,label=above:{$\ell_1$}](l1){};
		\draw (90-24*14:2.5) node[v,label={$\ell_2$}](l2){};
		
		\draw (0,0) node[v](v){};
		\foreach \x in {1,2,3,4,5,6,7,8,9,10,11,12,13,14}{
			\draw (v)--(v\x);
			\draw (v\x)--(l\x);
		}
		\draw (v14)--(w)--(v);
		\draw (w)--(v1);
		\draw (v2)--(u1)--(v4);
		\draw (v6)--(u2)--(v8);
		\draw (v10)--(u3)--(v12);
	\end{tikzpicture}
	\caption{Attaching a tent to a face bounded by a cycle $u_1u_2u_3u_1$.}
	\label{fig:Fgraph}
\end{figure}

\begin{PROP}\label{prop:planar4}
	For a $2$-connected plane graph $G$, let $\widetilde{G}$
	be a plane graph obtained from $G$ by attaching a tent to every face.
	Then $\pcf(G) \leq 3$ if and only if $\pcf(\widetilde{G}) \leq 4$.
\end{PROP}
\begin{proof}
	Suppose that~$G$ has a PCF $3$-coloring $c$.
	For a vertex $v$ of $\widetilde{G}$, let 
	\[
		c'(v):=
		\begin{cases}
			c(v)	& \text{if $v\in V(G)$},\\
			1	& \text{if $v=v^f$ for a face $f$},\\
			2	& \text{if $v=w^f$ or $v=\ell_i^f$ for a face $f$ and an integer $i$},\\
			3	& \text{if $v=v^f_{i}$ for a face $f$ and an odd integer $i$},\\
			4	&\text{if $v=v^f_{i}$ for a face $f$ and an even integer $i$}.
		\end{cases}
	\]
	We claim that $c'$ is a PCF $4$-coloring of $\widetilde{G}$.
	It is easy to see that $c'$ is a proper $4$-coloring.
	
	If $v$ is a vertex of~$G$,
	then every vertex in $N_{\widetilde{G}}(v)\setminus N_G(v)$ is colored by $4$, and therefore there is a color in $\{1,2,3\}$ appearing uniquely in the neighborhood of $v$ in $\widetilde{G}$ because $c$ is a PCF $3$-coloring of $G$.

	Since the neighborhood of $w^f$ or $\ell_i^f$ is a clique, every color in the neighborhood of $w^f$ or $\ell_i^f$ appears uniquely. 
	In the neighborhood of $v^f$, the color~$2$ of $w^f$ appears uniquely.

	If $v_i^f$ is adjacent to a vertex $v$ in $G$, then the color $1$ of $v^f$ or the color~$2$ of~$\ell_i^f$ appears uniquely in the neighborhood of $v_i^f$, because $v$ is the only vertex of $G$ adjacent to $v_i^f$.

	If $v_i^f$ is non-adjacent to any vertex of $G$, then the color $1$ of $v^f$ appears uniquely in the neighborhood of $v_i^f$.
	Therefore, $c'$ is a PCF $4$-coloring of~$\widetilde{G}$, proving the claim.
	
	Conversely, suppose that $\widetilde{G}$ has a PCF $4$-coloring $c$.
	We are going to show that the restriction $c'$ of $c$ on $V(G)$ is a PCF $3$-coloring of~$G$.

	First we claim that for every face $f$ of $G$, there is a color $\xi(f)$ such that if $v_i^f$ is adjacent to some vertex of $G$, then $c(v_i^f)=\xi(f)$.
	Let $f$ be a face and $k$ be the length of the boundary of $f$.
	Since $\{v^f,w^f,v_1^f, v_{4k+2}^f\}$ is a clique, all four colors appear inside this clique.
	Since $c$ is a PCF $4$-coloring, there is a color $i$ of some vertex in $\{w^f,v_1^f, v_{4k+2}^f\}$ appearing uniquely in the neighborhood of $v^f$.
	This implies that all vertices $v_i^f$ for $2\le i\le 4k+1$ avoid the color $i$ and the color of $v^f$.
	Since those vertices induce a connected bipartite graph and we only have two available colors, there is a unique bipartition, and therefore all vertices in $\{v_2^f,v_4^f,v_6^f,\ldots, v_{4k}^f\}$ have the same color, say $\xi(f)$.
	This proves the claim.

	Note that if a vertex $v$ of~$G$ incident with a face $f$, then $v$ has two neighbors in $\widetilde{G}$ inside $f$ and they have the same color $\xi(f)$.
	Since $c$ is a PCF $4$-coloring of $\widetilde{G}$, there must be a color appearing uniquely in the neighborhood of $v$ in $\widetilde{G}$, and therefore $c'$ is a PCF $4$-coloring of $G$.
	So, it remains to show that $c$ avoids one color.
	
	For that, we will show that 
	$\xi$ is a constant function. 
	It suffices to show that for each vertex $v$ of~$G$, all faces incident with $v$ have the same value of~$\xi$.
	Let $v$ be a vertex of $G$.
	Since there is a color $i$ appearing uniquely in the neighborhood of $v$ in $\widetilde{G}$, 
	for every face $f$ incident with $v$, 
	we have $\xi(f)\notin\{i, c(v)\}$.
	So, there are at most two distinct values of $\xi$ among faces incident with $v$.
	
	If $e=vw$ is an edge incident with faces $f_1$ and $f_2$ such that $\xi(f_1)\neq \xi(f_2)$, then 
	$c(w)$ is the unique element of $[4]\setminus \{ c(v), \xi(f_1),\xi(f_2)\}$
	because $c$ is a proper $4$-coloring of $\widetilde{G}$.
	So, if there are two distinct faces $f_1$ and $f_2$ incident with $v$ 
	such that $\xi(f_1)\neq \xi(f_2)$, 
	then there are two edges $vw_1$ and $vw_2$ incident with faces of distinct values of $\xi$ and so 
	$c(w_1)=c(w_2)$.
	This implies that no color appears uniquely in the neighborhood of $v$ in $\widetilde{G}$, contradicting the assumption that $c$ is a PCF $4$-coloring of $\widetilde{G}$.
	Therefore, $\xi$ is a constant function and this completes the proof.
\end{proof}

\section{Conclusion}\label{sec:end}
We showed that for each $k\geq3$, it is NP-complete to decide whether $\pcf(G)\le k$ and to decide whether $\och(G)\le k$ for bipartite graphs $G$, and also showed that it is NP-complete to decide whether $\pcf(G)\le 4$ for a planar graph $G$. 

On the other hand, it is straightfoward to see that for each fixed $k$ and~$t$, deciding whether $\pcf(G)\le k$ or whether $\och(G)\le k$ can be done in time $O(n^3)$ for an $n$-vertex graph of clique-width at most $t$.
This is because those problems are expressible in the counting monadic second-order logic, which can be decided for graphs of bounded clique-width (or rank-width); see \cite[Proposition 5.7]{CO2008}.

Fabrici, Lu\v{z}ar, Rindo\v{s}ov\'a, and Sot\'{a}k~\cite{Fabrici2022} showed that every planar graph admits a PCF $8$-coloring and conjectured that every planar graph admits a PCF $6$-coloring.
They presented a planar graph with the PCF chromatic number $6$.
Petru\v{s}evski and \v{S}krekovski \cite{Petrusevski2021} conjectured that every planar graph admits an odd $5$-coloring.
As related questions, we propose the following problems.

\begin{PROB}
	Determine the computational complexity of deciding whether a planar graph admits a PCF $k$-coloring for $k\in\{5,6,7\}$.
\end{PROB}

\begin{PROB}
	Determine the computational complexity of deciding whether a planar graph admits an odd $k$-coloring for $k\in\{4,5,6,7\}$.
\end{PROB}

\providecommand{\bysame}{\leavevmode\hbox to3em{\hrulefill}\thinspace}
\providecommand{\MR}{\relax\ifhmode\unskip\space\fi MR }
\providecommand{\MRhref}[2]{%
  \href{http://www.ams.org/mathscinet-getitem?mr=#1}{#2}
}
\providecommand{\href}[2]{#2}

\end{document}